\documentclass{article}
\usepackage{graphicx, amsmath, amsthm, amsfonts, hyperref, cleveref, enumerate, appendix} 
\usepackage[round,comma]{natbib}
\usepackage{tabularray}

\UseTblrLibrary{booktabs}
\UseTblrLibrary{siunitx}
\usepackage{threeparttablex}
\usepackage{booktabs}
\usepackage{threeparttable}

\usepackage{siunitx}
\usepackage{setspace}
\usepackage{geometry}
\usepackage{float}
\usepackage{codehigh}
\usepackage[normalem]{ulem}

\NewTableCommand{\tinytableDefineColor}[3]{\definecolor{#1}{#2}{#3}}

\newcommand{\sym}[1]{\ifmmode^{#1}\else\(^{#1}\)\fi}

\newtheorem{lemma}{Lemma}

\newtheorem{proposition}{Proposition}

\title{Attention and Social Learning\thanks{We thank Cuimin Ba, Drew Fudenberg, Marina Halac, Jawwad Noor, Doron Ravid, Collin Raymond and seminar and conference audiences for valuable comments. We thank Raghu Arghal, Varshitha Satish Kumar, Stephan Xie, and Kyle Xiong for excellent research assistance on this project or earlier work incorporated into our experimental platform. We gratefully acknowledge financial support from NSF Grants SES-2214950 and SES-2215256.}}
\author{Krishna Dasaratha\thanks{Boston University. Email: \texttt{\protect\href{mailto:krishnadasaratha@gmail.com}{krishnadasaratha@gmail.com}}}
\and Kevin He\thanks{University of Pennsylvania. Email: \texttt{\protect\href{mailto:hesichao@gmail.com}{hesichao@gmail.com}}}}
\date{\today}

\begin{document}

\maketitle

\begin{abstract}

In an incentivized laboratory experiment, we study how people account for and respond to others' incentives for paying attention. Participants learn a binary state from an attention task under high or low accuracy incentives. We ask  subjects to predict their peers' accuracy based on the peers' incentives and to aggregate answers from multiple peers with different incentives. Most subjects fail to consistently understand that peers with stronger incentives are more accurate, and these subjects  also perform worse in individual attention tasks. Subjects also participate in a social-learning task where they  first learn the binary state from an attention task, then observe a peer's guess about the state in the same task, and finally make a guess themselves. We find behavior in these tasks is inconsistent with leading models of flexible costly information acquisition. In particular, subjects fail to pay more attention when paired with lower incentive peers. Overall, we find that many decision-makers do not respond to others' incentives for accuracy even when those incentives are transparent. 
\end{abstract}
\newpage
\section{Introduction}

People often learn from peers who only pay partial attention to the information available to them. For instance, we mostly learn about current events through news reporting produced by journalists, who often have more direct access to relevant information sources than the general public.  Given the wealth of information available to journalists, they must decide which parts to focus on and how much attention they should pay. This may take the form of deciding how many experts to interview or how carefully to fact-check before publishing their work. Similarly,  households learn about investment opportunities through their friends paying attention to market data, patients learn about diagnoses and treatment options through medical professionals paying attention to test results and clinical guidelines, and shoppers learn about products through the recommendations of  reviewers paying attention to product features and performance benchmarks. In all of these settings, people rely on intermediaries who have exerted costly effort to attend to information. 

The attention levels of these information intermediaries depend on the incentives they face. For example, a journalist who works for a newspaper that heavily rewards its employees for careful and accurate coverage will likely exert more attentional effort verifying facts than their counterpart at a media outlet oriented toward sensationalist reporting. This leads to the main question of our study: how well do people learn from others in situations involving attention, and in particular do people correctly understand how variations in incentives to pay attention affect the accuracy of social information from different peers?

There are several plausible channels through which attention could limit or bias social learning. For example, if people neglect the effect of incentives on others' accuracy, they may trust unreliable sources and reliable sources equally. Another possibility is that people rely too heavily on others' attention, which could lead to inefficiently low levels of attention in larger groups (see \cite{kubler2004limited}). As one real-world example, \cite{sundar2025sharing} find that when Facebook users re-share others' posts that include a link, only 25\% of them click on the link before re-sharing -- with many of the remaining 75\% presumably relying on the original poster to have paid attention to the linked story and verified its quality.

We study how people account for and respond to others’ incentives for paying attention in an incentivized laboratory experiment, using the online experimental subject pool Prolific. A first wave of 600 participants guesses an unknown binary state, which they can learn perfectly by counting balls of two different colors on their screen (as in \cite{dean2023experimental}). Subjects are randomly assigned to an incentive-level treatment that determines how much they are paid when they guess correctly. We find that subjects guess with 69\% accuracy in a high-incentive treatment and 63\% accuracy in a low-incentive treatment. 

The main focus of the study is on how a second wave of 393 participants infers the state from the guesses and incentive levels of their first-wave predecessors. Subjects in the second wave complete three types of tasks. First, we ask second-wave subjects to estimate the average accuracies of the first-wave participants in the two incentive levels. Second, we ask the second-wave subjects to form a belief about the state given the guesses of two first-wave participants, one from each incentive level. Third, we ask the second-wave subjects to participate in a task involving both paying attention themselves and learning from others, which we call the \emph{attention substitution task}. Subjects  face the  same attention task where the state is revealed by the color majority on their screen. After they have paid attention to the screen and before they make their final guess, these subjects additionally observe the binary guess of a first-wave participant in the same attention task and the incentive level of this predecessor. Second-wave subjects can combine this social observation with the information they extracted from paying attention themselves when making a final guess about the state.

Results from the first two tasks show that around half of the subjects fail to infer even the \emph{directional} effect that higher incentives for attention generate higher accuracy in their peers.  In line with a well-documented bias in social-learning experiments  \citep{weizsacker2010we, conlon2022not}, second-wave subjects underestimate  first-wave subjects' accuracy at every incentive level. Subjects predict an average accuracy of 57\% in the low-incentive treatment and 60\% in the high-incentive treatment. These aggregate beliefs mask a wide heterogeneity, where  only a slight majority (55\%) of subjects estimate a strictly higher accuracy in the high-incentive treatment than in the low-incentive treatment. The other subjects either predict the same accuracy in each (19\%) or lower accuracy in the high-incentive treatment (25\%). We find similar patterns in the second task, when second-wave subjects are asked to aggregate two guesses made by two predecessors with different incentive levels.  When the two predecessors' guesses disagree, subjects correctly place a higher weight on the high-incentive guess 42\% of the time. The remaining majority of the subjects either treat the two guesses as equally informative or treat the low-incentive guess as more informative. Overall, we find that most subjects fail to consistently understand how incentives for attention affect peers' accuracy.

We can explain part of this heterogeneity based on performance in individual tasks: subjects who are more accurate in individual attention tasks are more likely to treat peers with stronger incentives as more accurate. There is not a similar association between  time spent on individual attention tasks and treating stronger incentive peers as more accurate, suggesting the heterogeneity may be due to differences in ability to process information rather than differences in the amount of effort that people put into completing the experiment. The upshot is that there is a large group of subjects who are both less accurate in processing information individually and less able to discern the relative accuracy of social information from different sources based on the sources' incentives. 

In the attention-substitution task, we similarly find that subjects' attention behavior is not responsive to others' incentives for attention. As outlined above, in this task subjects are given a chance to learn about the binary state by paying attention. After completing the attentional task, each subject observes the binary guess of a predecessor (with a known incentive level) and makes a binary guess about the state. This task lets us ask a more applied question about how subjects' own attention depends on others' incentives, but is also designed to test a standard class of costly information acquisition models. Suppose people are rational and have a posterior-separable cost of  information  acquisition satisfying mild conditions. One example is the rational inattention model with Shannon entropy cost. Another example is the drift-diffusion model, where decision-makers pay attention by watching a noisy continuous-time process with a state-dependent drift and incur a constant cost per unit of time until they decide to stop. This class of models predicts that  subjects will pay (weakly) more attention when they anticipate receiving social information from a less incentivized (and thus less accurate) peer, but this prediction is not consistent with the data. We consider two measures of attention: how long subjects spend on paying attention and learning the state, and accuracy in an unincentivized initial guess based only on their own attention. We do not find that subjects pay more attention by either of these two measures when they anticipate getting social information from a predecessor with a lower incentive level. In fact, these two measures of attention are directionally \emph{lower} for subjects paired with low-incentive predecessors, and statistically significantly so for the amount of time spent on attention. 

Another theoretical prediction of the posterior-separable cost of information acquisition is that subjects' choices satisfy a \emph{no integration} condition: if they pay any non-zero attention cost, then they must rationally ignore their social information. There is never non-trivial integration of the information acquired through attention and the information from social peers. We find that subjects pay substantial attention in ways that violate this  prediction. On average, the level of attention among second-wave subjects who anticipate  receiving social information is roughly comparable to that of first-wave subjects in the high-incentive treatment, despite the fact that the latter were paid more for correct guesses and had no access to social information. Even though subjects pay a lot of attention, they do not ignore social information. Among subjects who spend at least 10 seconds counting the colored balls, 20\% change their initial guess when it is contradicted by their predecessor's guess. Overall, subjects' attentional choices either are not fully explained by rational models or require relaxing posterior-separability of information costs or full flexibility of information choice. 

In summary, we find that subjects often do not respond to others' incentives for attention. First, only a slight majority of the subjects predict accuracy is strictly higher in the high-incentive treatment. Second, only a minority of subjects treat high-incentive guesses as  strictly more informative than low-incentive guesses in  aggregation tasks. Third, subjects do not pay more attention themselves when they anticipate getting social information from a lower-attention information source. These findings indicate that even in abstract settings where others' incentives for accuracy are transparent, many people are not responsive to these incentives. Beyond experimental settings, such non-responsiveness would lead people to ignore the differences in reliability between information sources and to independently fact-check sources with high-incentives for accuracy as much as those with low-incentive.

\section{Related Literature}

Our experiment relates to an experimental literature on incentives for attention. Most recent papers look at single-agent settings (e.g., \cite{caplin2020rational}, \cite{dean2023experimental}, and \cite{bronchetti2023attention}). Phase One of our experiment is based on this literature, but our main focus is on how incentives for attention impact social learning in multi-agent settings. Our attention-substitution task also generates indirect tests of models of costly information acquisition, which we see as complementing more direct tests in prior work. In particular, \cite{dean2023experimental} find subjects' behavior is consistent with rational models of costly attention but not with the specific functional form of entropy cost. \cite{denti2022posterior} finds that \cite{dean2023experimental}'s subjects also violate axioms satisfied by any rational agents with posterior-separable cost of attention. We test different predictions of the posterior-separable model and also find several violations.

In multi-agent settings, \cite{almog2024rational} and \cite{spurlino2025rationally}  study trading games where subjects can learn about an object's value through attentional tasks. \cite{almog2024rational} pair perfectly informed sellers with buyers who choose how much attention to pay to the object's value. Their focus is on buyer behavior rather than seller beliefs about buyer attention, which are not elicited. \cite{spurlino2025rationally} implement a trading game where both agents can pay attention to the value of a trade and separately vary the difficulty of each agent's attentional task. Beliefs about others' attention are inferred from agent actions rather than directly elicited. A main finding is that most subjects do not respond to the difficulty of the other agent's attention task, paralleling our result that a majority of subjects fail to treat higher-incentive subjects as more informative. 

In a social-learning setting, \cite{kubler2004limited} study sequential social learning when agents face monetary (and not attentional) costs to acquire private signals. The Nash equilibrium prediction is that only the first agent acquires a signal and all others mimic their action. But the experiment found that many subjects buy  signals, including when doing so is suboptimal; we similarly find that subjects pay large amounts of attention to learn about the state themselves even when they can learn from others.

Finally, our tasks in which subjects aggregate multiple peers' guesses relate to a large literature on belief updating (as surveyed by \cite{benjamin2019errors}). Recent work in settings where the complexity of the updating task presents a challenge to subjects includes \cite*{augenblick2025overinference} and \cite*{ba2024over}. Whereas this literature focuses on errors in information processing and belief updating, our paper focuses on a complementary bias where people misunderstand how economic incentives of the information intermediaries affect their accuracies. Indeed, we find in the aggregation tasks that subjects' average beliefs about the state are close to the correct Bayesian posteriors given their average beliefs about peers' guess accuracies. But, most people fail to appreciate how others' incentives affect their accuracies, leading to \emph{directional} errors when reconciling contradictory views from multiple peers or deciding how much effort to put into verifying the claims of different sources.  

\section{Experimental Setup}

A total of 993 subjects participated in the study through the  online platform Prolific in February and March 2025. The experiment had two phases. Phase One involved 600 subjects who repeatedly performed an attention task. Phase Two involved 393 subjects who completed a combination of attention tasks, estimation tasks about others' accuracy, and attention-substitution tasks. We  recruited subjects located in the United States who had  completed at least 10 prior studies on Prolific and had an approval rate of at least 95\%. Subjects were paid both a fixed show-up fee and a bonus payment that depended on the accuracy of their responses.

Before starting the experiment, all subjects answered comprehension questions to ensure that they understood the instructions.  As discussed below, we also included additional comprehension questions halfway through the study for Phase Two subjects to test whether they understood the timing and payment rules in  the attention-substitution task.  Subjects had two tries to correctly answer all the comprehension questions; otherwise, they were screened out of the study and replaced with new subjects.\footnote{A total of 400  participants completed Phase Two of the study on Prolific, and we  excluded the data from seven of these subjects because they were screened out by comprehension questions or exceeded the time limit but told the Prolific system that they completed the study successfully.}

The experimental instructions and interface  can be found in Appendix \ref{sec:interface}. We pre-registered this study, including the attention and social-learning tasks, the incentive levels, the main regressions, and the expected sample sizes (which were almost exactly met). The pre-registration document can be found at: \url{https://aspredicted.org/h3mc7s.pdf}. The remainder of this section describes the tasks completed by subjects in each phase of the experiment.

\subsection{Phase One}

Each subject in Phase One completed 20 repetitions of an attention task. The goal in each repetition was to guess whether a hidden fair coin landed heads or tails. The screen showed a $15 \times 15$ grid of 225 balls that were either orange or purple  (see \Cref{fig:grid_example}). Subjects knew  that there would be 115 orange balls and 110 purple balls in the grid if the coin landed heads, and 115 purple balls and 110 orange balls in the grid if the coin landed tails. Subjects made a binary guess about the outcome of the coin toss based on the grid.  The fair coin was re-tossed for each  of the 20 repetitions of the attention task and a new grid was generated based on  the new coin-toss outcome. Subjects could spend as much or as little time as they chose on each repetition.\footnote{Subjects timed out of the study if they did not complete it within 75 minutes, but our data show this was not a binding time constraint for any of the subjects.} 

At the time of our experiment, large language models including ChatGPT o3-mini-high (which was the most advanced model available under OpenAI's Plus subscription plan) struggled with this problem. Moreover, we found that they frequently gave clearly nonsensical answers or took several minutes to output a response.

Subjects were paid a bonus payment at the end of the experiment based on how many correct guesses they made in the 20 repetitions, but received no feedback about their accuracy between the repetitions. The participants in Phase One of the experiment were randomly assigned into either the high-incentive treatment or the low-incentive treatment, with $300$ subjects in each. Subjects were told at the start of the study how much they would get paid per correct guess, which was \$1.25 for those in the high-incentive treatment and \$0.05 for those in the low-incentive treatment.

\begin{figure}
    \centering
    \includegraphics[scale=.25]{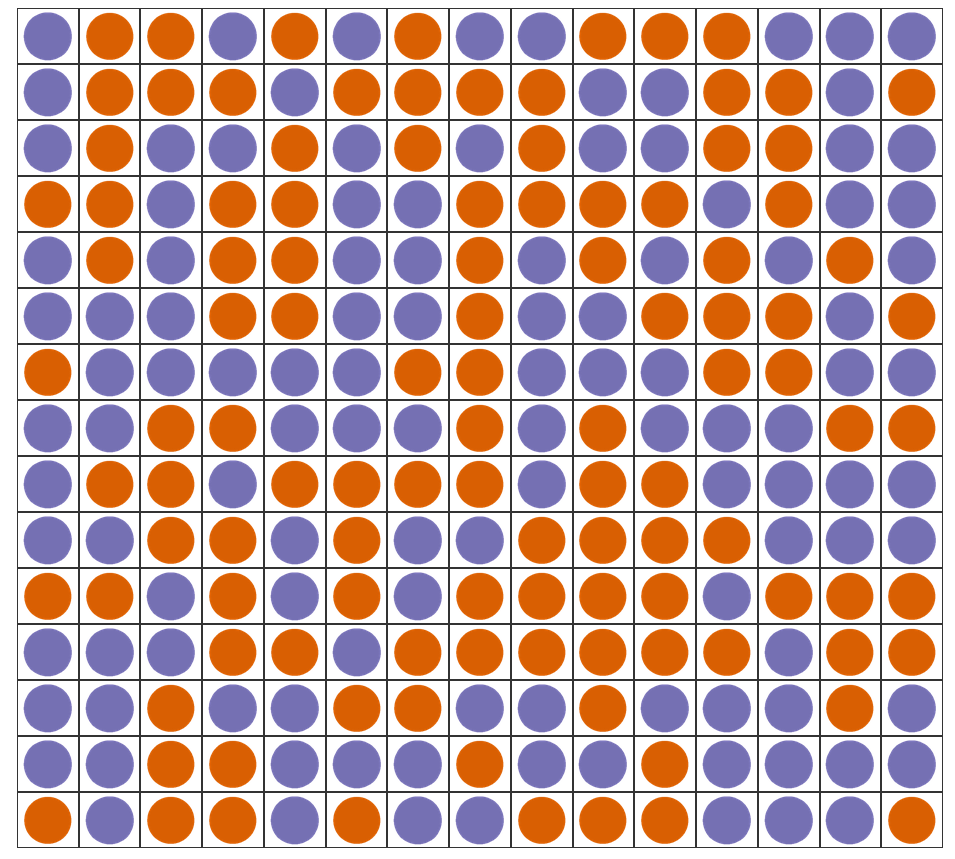}
    \caption{An example grid with 115 purple balls and 110 orange balls.}
    \label{fig:grid_example}
\end{figure}

\subsection{Phase Two}

Participants in Phase Two of the study completed a series of tasks that involved estimating the accuracy of Phase One subjects' guesses or making inferences about the outcomes of the coin tosses based on these guesses. Each subject was initially randomly assigned to a high-incentive or low-incentive treatment as in Phase One, with 198 subjects in the low-incentive treatment and 195 subjects in the high-incentive treatment. Phase Two subjects were eventually told that some Phase One participants had an incentive level different from theirs (as described below), but we postponed providing this information to minimize experimenter demand effects on subjects' estimations of accuracies under different incentive levels. We now describe the tasks that Phase Two subjects completed in the  order that they appeared in the experiment:

\begin{enumerate}
\item \textbf{Attention Task:} Each subject began by completing two repetitions of the attention task from Phase One to gain familiarity. Subjects in the high-incentive treatment were paid \$1.25 for each correct guess, while those in the low-incentive treatment were paid \$0.05 for each correct guess.

\item \textbf{Beliefs about Accuracy in Own Treatment:} Subjects were told that some Phase One subjects completed attention tasks under the same incentive structure as them. They were then asked to make a point estimate about the average accuracy among the guesses of these earlier peers. We paid subjects \$0.50 if the estimate  was within 2 percentage points of the truth. 

\item \textbf{Attention-Substitution Task:} Subjects then completed 10 repetitions of a task where they first decided how much  attention to pay to  learn about a coin-toss outcome,  then  observed the guess of a Phase One subject about the same outcome. That is, subjects made attention choices in anticipation of later observing social information in the form of another  subject's guess. The timeline of each repetition of the task was as follows:
\begin{enumerate}[(a)]
    \item A hidden fair coin was tossed and the screen showed a grid generated in the same way as in the attention task from Phase One. The subject could freely choose how much time to spend on paying attention to this grid. 
    \item The subject made an unincentivized binary guess about the coin-toss outcome.
    \item The grid was removed from the screen and the subject learned the binary guess of a Phase One subject who saw a grid generated based on the same coin toss outcome (the placements of the colored balls in the grids were independently randomized for the Phase One and Phase Two subjects). A Phase Two subject in the low-incentive treatment always saw the guess of a Phase One peer from the low-incentive treatment, and vice versa for a Phase Two subject in the high-incentive treatment. 
    \item The subject made a final binary guess of the coin-toss outcome with a \$0.50 bonus for correctness. (This bonus amount did not vary based on the Phase Two subject's incentive treatment.) The subject's final  guess could incorporate information that they  acquire from paying attention to the grid on the previous screen  and/or information contained in  the Phase One peer's guess. (It was not possible to return to the previous screen to see the grid again.) 
\end{enumerate}

This sequence of events was fully explained to subjects before they began the first repetition of the task. In particular, subjects were told that they would be observing a Phase One subject's guess and were told how much that Phase One subject was paid for correct guesses. After this explanation, subjects were required to complete additional comprehension questions to ensure that they understood the timeline of the task and how their bonus payment for this task was determined. 

\item \textbf{Beliefs about Accuracy in Other Treatment:} At this point in the experiment, we revealed for the first time that some Phase One subjects were paid \$0.05 per correct guess, while others were paid \$1.25 per correct guess. The Phase Two subjects were then asked to make a point estimate about the accuracy of Phase One subjects in the incentive treatment that was different from their own. We again paid a \$0.50 bonus if the estimate was within 2 percentage points of the truth.

\item \textbf{Aggregating Guesses from Different Incentive Treatments:} In the final task, we asked Phase Two subjects to aggregate two guesses made by two Phase One subjects from different incentive treatments about  the same coin-toss outcome.\footnote{To set up this task, during Phase One of the study we paired high-incentive and low-incentive subjects and had them make guesses about the same coin-flip outcome based on conditionally independent grids.} For example, one aggregation task involved situations where a high-incentive Phase One subject guessed ``heads'' while a low-incentive Phase One subject guessed ``tails'' about the same coin-toss outcome.  We then asked the Phase Two subject to estimate the fraction of such situations in which the true coin-flip outcome was heads. We paid a \$0.50 bonus if the estimate was within 2 percentage points of the truth.

Each Phase Two subject made four estimates, corresponding to the four possible combinations of binary guesses of the high-incentive and low-incentive peers (HH, HT, TH, and TT). We randomized the order in which the four combinations were presented and whether the first guess was made by a high-incentive subject or a low-incentive subject within each combination.

\end{enumerate}

\section{Accuracy, Beliefs about Accuracy, and Aggregation Tasks}

This section presents Phase One results on subjects' performance in attention tasks and Phase Two results on subjects' beliefs about others' accuracies under different incentive treatments, as revealed through direct elicitation and through aggregation tasks. 

\subsection{Performance in Attention Tasks}

We begin by comparing accuracy in the two incentive treatments. We find that 63\% of low-incentive guesses are correct and 69\% of high-incentive guesses are correct. To establish significance of this difference, we consider the pre-registered regression 

\begin{equation}
A_{i,j}=\alpha+\beta_{H}\cdot H_{i}+\epsilon_{i,j},\label{eq:accuracy_basic}
\end{equation}
where $A_{i,j}$ is a binary variable indicating whether subject $i$
correctly guessed the coin-toss outcome in the $j$-th round of the
attention task for $1\le j\le20$, $H_{i}=1$ if subject $i$ is in
the high-incentive treatment and $H_{i}=0$ otherwise. We cluster
error terms $\epsilon_{i,j}$ at the subject level in all applicable
regressions. We estimate (see Table \ref{tab:accuracy_five_specs}, Column (1)) $\hat{\alpha}=0.629$
and $\hat{\beta}_{H}=0.060$, both significant at the 1\% level. As
in \cite{dean2023experimental}, we find that Phase One subjects are responsive to incentive
levels. 

High-incentive subjects also spend longer on the attention task. To see this, we estimate the regression
\[
D_{i,j}=\alpha+\gamma_{H}\cdot H_{i}+\gamma_{R}\cdot j+\epsilon_{i,j}
\]
 where $D_{i,j}$ is the number of seconds that subject $i$ spent
in the $j$-th round of the attention task for $1\le j\le20$ (see Table \ref{tab:reg_duration} for results). We
find $\hat{\alpha}=28.27,$ $\hat{\gamma}_{H}=12.07$, and $\hat{\gamma}_{R}=-0.95$,
all significant at the 1\% level. On average high-incentive subjects
spend 12 more seconds on each attention task, and subjects tend to
spend about 1 second less on each subsequent repetition of the task.

Spending more time on the attention task is associated with higher accuracy. We can see this from our preferred specification of the accuracy regression (Table
\ref{tab:accuracy_five_specs}, Column (3)): 
\[
A_{i,j}=\alpha+\beta_{H}\cdot H_{i}+\beta_{D}\cdot D_{i,j}+\beta_{R}\cdot j+\epsilon_{i,j}.
\]
We estimate $\hat{\beta}_{D}=0.0032$ and $\hat{\beta}_{R}=0.0032$,
both significant at the 1\% level. Spending ten more seconds on an
attention task is associated with a 3.2 percentage points increase
in accuracy, which is the same effect size as gaining experience equivalent to ten repetitions of the attention task. This suggests that a substantial portion of the accuracy difference between treatments is explained by the difference in time spent. We also find $\hat{\beta}_{H}=0.020$
with a $p$-value of 0.073. This shows that higher incentive is weakly
associated with higher accuracy even controlling for the length of
time spent, which may come from greater intensity of attention or
mental effort.

\subsection{Beliefs about Others' Accuracy}\label{s:beliefs_about_accuracy}

We next describe Phase Two subjects' estimations of Phase One subjects' accuracy. We begin by analyzing average beliefs about accuracy and then turn to exploring heterogeneity.

Recall that  Phase Two subjects were initially only aware of one incentive treatment and were asked to estimate their Phase One peers' average accuracy at the same incentive level as their own. They were later told about the other incentive treatment and asked to estimate earlier subjects' average accuracy at this other incentive level. Subjects spent an average time of 32 seconds on each of these estimates.

Our pre-registration includes a main regression where we only use each subject's initial estimation about accuracy in their own treatment (without knowledge of the other treatment) and a secondary regression where we use both estimations from each subject. For the first regression, we have

\begin{equation}
E_{i,\text{own}}=\alpha+\beta_{H}\cdot H_{i}+\epsilon_{i},\label{eq:estimate_own}
\end{equation}
where $E_{i,\text{own}}$ refers to $i$'s estimate about average
accuracy of the Phase One subjects who had the same incentive level
as $i.$ Here $H_{i}\in\{0,1\}$ indicates whether $i$ is assigned
to the high-incentive treatment. 

For the second regression, we have

\begin{equation}
E_{i,j}=\alpha+\beta_{H}\cdot H_{i,j}+\epsilon_{i,j}\label{eq:estimate_all}
\end{equation}
 for $j\in\{\text{own, other}\}$, where $E_{i,\text{other}}$ is
$i$'s estimate about average accuracy of the Phase One subjects who
have a different incentive level than $i.$ Here $H_{i,\text{own}}\in\{0,1\}$
indicates whether $i$ is assigned to the high-incentive treatment
and $H_{i,\text{other}}=1-H_{i,\text{own}}.$ 

The regression results can be found in Table \ref{tab:point_estimates}. Among initial estimations, the average estimates of accuracy are 56.7\% for the low-incentive treatment and 60.4\% for the high-incentive treatment. Among all estimations, the average estimates of accuracy are 56.2\% for the low-incentive treatment and 62.0\% for the high-incentive treatment. Recall the true accuracies are 62.9\% for the low-incentive treatment and 68.9\% for the high-incentive treatment.

Subjects underestimate others' accuracy in both treatments. This is consistent with a well-documented bias of underweighted information from others in a social learning setting, which we discuss in Section \ref{s:aggregation}. Subjects' average perceptions of the effect size of  incentive on accuracy are also directionally too low. They believe this effect size to be 3.7 percentage points among the initial Phase Two estimates and 5.7 percentage points among all Phase Two estimates, but the true effect size is 6.0 percentage points. However, the underestimation of the effect size of higher incentive is not statistically significant. We conduct our pre-registered Wald test of comparing the coefficient estimate of $\beta_H$ from Equation (\ref{eq:accuracy_basic}) with the coefficient estimates of $\beta_H$ from Equations (\ref{eq:estimate_own}) and (\ref{eq:estimate_all}). We cannot reject the hypothesis that the coefficients are the same for either of these two comparisons at the 10\% level. This is partly due to the considerable heterogeneity in Phase Two subjects' beliefs about Phase One peers' accuracy, as we will discuss shortly.

We can also compare the actual and perceived informativeness of high-incentive and low-incentive guesses, which we compute by comparing log-likelihood ratios.\footnote{We define the ratio of informativeness between a signal with accuracy $q$ and a signal with accuracy $q'$ as $\frac{\log(q/(1-q))}{\log(q'/(1-q'))}$.} The actual ratio of informativeness between high-incentive and low-incentive guesses is 1.50. The averages across initial estimates of Phase Two subjects correspond to a ratio of informativeness of 1.57, while the averages across all  estimates correspond to a somewhat higher ratio of 1.95.

Underlying these average beliefs about how higher incentive affects accuracy, which are either not statistically or not substantially different from the truth, we find a wide heterogeneity in beliefs. To examine this heterogeneity, we focus on the distribution of the difference between beliefs about high-incentive accuracy and low-incentive accuracy across Phase Two subjects. It is worth pointing out that Phase Two subjects completed a series of other tasks in between  making the two estimates about accuracy under two incentive levels,  which may limit consistency across an individual's two estimates.

\begin{figure}
    \centering
    \includegraphics[width=0.7\linewidth]{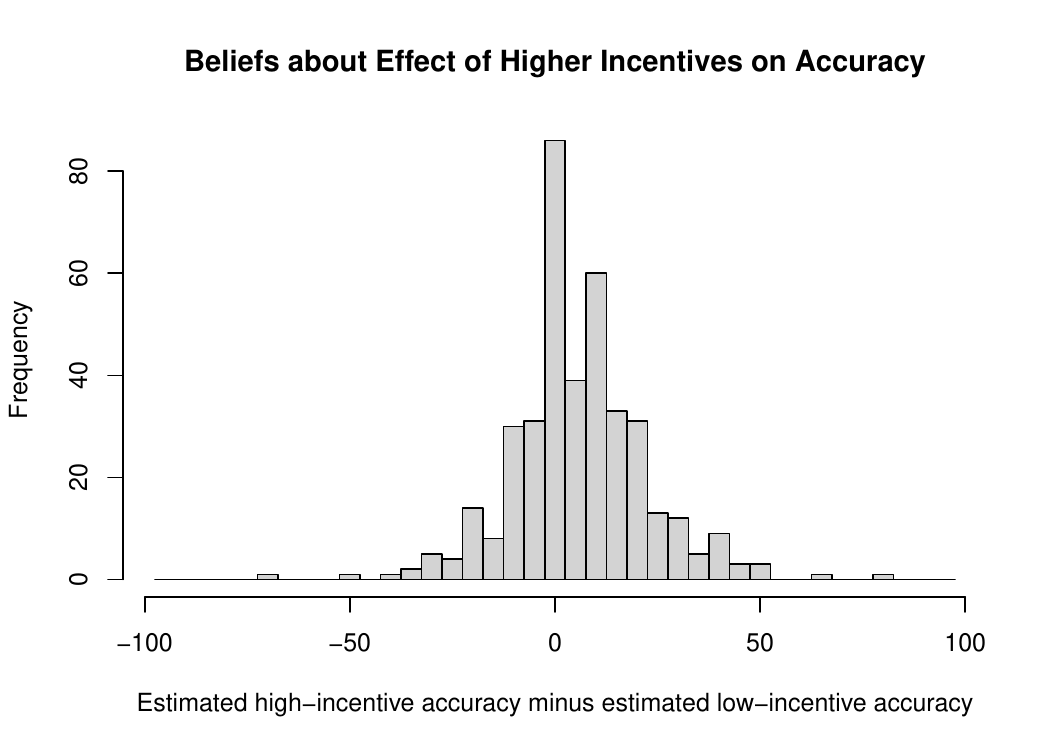}
    \caption{Histogram of the within-subject differences in Phase Two subjects' estimates about high-incentive accuracy and low-incentive accuracy in the attention task. The difference is shown in percentage points.}
    \label{fig:hist_acc_belief}
\end{figure}

\Cref{fig:hist_acc_belief} plots the histogram of these differences in percentage points. The distribution has a standard deviation of 16.5, indicating many subjects predict that guesses in the high-incentive treatment are much more accurate than those in the low-incentive treatment \emph{and} many subjects predict the opposite. Only a slight majority of subjects (55.2\%) correctly predict the \emph{sign} of the treatment effect and estimate a positive accuracy difference between the high-incentive and low-incentive treatments. We find  19.3\% of the subjects predict exactly the same accuracy in both treatments, and 25.4\% predict a negative accuracy difference between the high-incentive and low-incentive treatments.

We next turn to Phase Two subjects' aggregation of  guesses from multiple peers and find similar heterogeneity in beliefs about how accuracy differs across incentive treatments.

\subsection{Aggregation of Guesses from Different Incentive Treatments}\label{s:aggregation}

We next describe how Phase Two subjects aggregate Phase One subjects' guesses. Recall that Phase Two subjects are asked to estimate the probability of a coin-toss outcome being heads given the guesses of two Phase One subjects, one from the low-incentive treatment and one from the high-incentive treatment,  who saw grids based on that coin toss. Each Phase Two subject makes a prediction for each possible pair of low-incentive and high-incentive guesses (HH, HT, TH, and TT). Subjects spent an average of 86 seconds on these predictions.

We use responses in  these aggregation tasks to back out Phase Two subjects' average beliefs about Phase One subjects' guess accuracies under different incentive treatments. To motivate our regression specification, suppose that low-incentive guesses are correct with probability  $q_L$ and high-incentive guesses are correct with probability   $q_H$. Let $x_L$ equal $1$ if the low-incentive guess is heads and $-1$ if the low-incentive guess is tails, and define $x_H$ analogously for the high-incentive guess. Then Bayes' rule requires that the posterior probability $p$ of heads after seeing these two guesses satisfies 
$$\log\left(\frac{p}{1-p}\right) = x_L \log\left(\frac{q_L}{1-q_L}\right) + x_H \log\left(\frac{q_H}{1-q_H}\right).$$
We can use the linearity to recover average beliefs about guess accuracies. That is, we consider the following pre-registered regression:
\[
\text{LogitBelief}_{i,j}=\lambda_{L}\cdot x_{L,j}+\lambda_{H}\cdot x_{H,j}+\epsilon_{i,j}
\]
where we number the possible combinations of low-incentive and high-incentive
guesses (HH, HT, TH, TT) as $j=1,2,3,4$. $\text{LogitBelief}_{i,j}$
is the logit function $\log(\frac{p_{i,j}}{1-p_{i,j}}),$ where $p_{i,j}$ is the subject $i$'s stated probability 
that the coin toss being heads after seeing the $j$-th combination
of low-incentive and high-incentive guesses.
The regressors $(x_{H,j},x_{L,j})$ record the two guesses observed in each combination $j$.  The regression estimates for $\lambda_L$ and $\lambda_H$ are our estimators for Phase Two subjects' beliefs about  $\log\left(\frac{q_L}{1-q_L}\right)$ and $\log\left(\frac{q_H}{1-q_H}\right)$. This exercise assumes that subjects are Bayesian, and we note that we are unable to distinguish misspecified beliefs about accuracies from non-Bayesian updating that over-weights or under-weights guesses.

The regression results are in Table \ref{tab:log_belief_regression}. We estimate log-likelihood ratios of $\hat{\lambda}_L = 0.257$ for low-incentive guesses and $\hat{\lambda}_H = 0.489$ for high-incentive guesses, which correspond to precisions $\hat{q}_L = 0.564$ and $\hat{q}_H = 0.620$. These implied beliefs from Phase Two subjects' aggregation task responses are remarkably close to their directly elicited average beliefs about others' accuracies, which were  $56.2\%$ and $62.0\%$ for the low-incentive and high-incentive peers (see \Cref{s:beliefs_about_accuracy}). So aggregate updating patterns are consistent with subjects using Bayes' rule given their reported beliefs about the accuracies of others' guesses. Phase One subjects' true accuracies in the two treatments ($62.9\%$ and $68.9\%$) are both higher compared to the respective implied beliefs of the Phase Two subjects, with the difference being statistically  significant at the 1\% level using the non-linear Wald test with delta method. But the true treatment effect size of increasing incentives matches the subjects' implied beliefs about the difference in accuracies between the two treatments. 

The implied beliefs reveal that Phase Two subjects place too little weight on both low-incentive and high-incentive peers' guesses. This matches prior experiments on social learning finding subjects place too little weight on information from others (including \cite{weizsacker2010we}, \cite*{conlon2022not}, and \cite*{agranov2025behavioral}).

As in  \Cref{s:beliefs_about_accuracy}, however, we find a large heterogeneity in responses and identify many subjects who fail to even directionally put more weight on high-incentive guesses than low-incentive guesses. Each subject completes two belief-aggregation tasks where the low-incentive and high-incentive guess disagree (HT and TH). We can ask how many of these subjects place more weight on the high-incentive guess in these tasks by comparing their estimations of the heads probabilities to $50\%$. In one such task, we find that $39.4\%$ of subjects place more weight on the high-incentive guess, $33.1\%$ of subjects place equal weight on each guess, and the remaining $27.5\%$ of subjects  place more weight on the low-incentive guess. In the other task (which is symmetric), we find $43.8\%$   place more weight on the high-incentive guess, $34.5\%$  place equal weight, and $21.7\%$ place more weight on the low-incentive guess. Only $21.6\%$ of subjects place (strictly) more weight on the high-incentive guess in both tasks.

\subsection{Who Understands the Effects of Incentives?}

In Sections~\ref{s:beliefs_about_accuracy} and \ref{s:aggregation}, we saw that large fractions of subjects failed to understand even the directional effect of others' incentives for accuracy. We now ask which subjects are making these errors, and find these subjects perform worse on individual accuracy tasks.

We compare understanding of incentive effects with performance on individual attention tasks. Are subjects who are more accurate in individual attention tasks more likely to treat higher-incentive peers as more accurate? To create a measure of individual accuracy, we compute average accuracy across the two trial attention tasks at the start of the study and in the ten initial guesses from the  Attention-Substitution Task. Table \ref{tab:who_understands} regresses two binary variables on this information measure: an indicator for estimating a strictly higher accuracy in the higher incentive treatment and an indicator for putting strictly higher weight on the higher-incentive guess in both aggregation tasks where the guesses disagree (i.e., HT and TH). Higher accuracy in individual tasks predicts understanding the directional effect of incentives in both tasks ($p$-value $<.01$ for accuracy estimates and $<.05$ for information aggregation).

Higher accuracy in individual tasks likely captures effort or ability in the experiment. Toward disentangling these stories, we ask whether failing to understand incentive effects is associated with spending less time. Table \ref{tab:who_understands} shows there is not a significant correlation between time spent on individual attention tasks and understanding the sign of incentive effects. There is also not a significant correlation with time spent on the belief about accuracy  and information aggregation questions; the subjects who treat high-incentive guesses as equally or less informative are still spending substantial time assessing others' accuracies and combining others' guesses. While we cannot directly observe how this time is spent in an online experiment, the subjects who fail to understand incentive effects are spending quite a bit of time throughout the experiment. This suggests the heterogeneity in understanding incentive effects is at least partially driven by differences in ability to process information. Under this interpretation, subjects who are more able to process information tend to correctly perceive the directional effect of incentives, while subjects who are less able to process information fail to understand these effects.

Overall, we find a group of almost half of subjects who perform less well at individual learning and do not treat external sources with stronger incentives for accuracy as more reliable. We highlight two consequences of this. First, heterogeneity in estimates of the effects of incentives likely reflect genuine differences between subgroups. Second, subjects who do not predict that information sources with stronger incentives for accuracy are indeed more accurate are also worse at directly processing information. This correlation could lead to particularly poor learning outcomes for these agents in settings involving both individual and social learning.


 
\section{Combining Attention and Social Learning}

We now analyze the Attention-Substitution Task, in which subjects first pay attention to a grid to learn about a coin-flip outcome, then observe another subject's guess about the same coin flip, and finally make an incentivized binary guess. We begin by deriving theoretical predictions of leading models of attention choice in this setting, and subsequently describe our experimental results.

\subsection{Theoretical Predictions of Posterior Separable Attention Cost}\label{s:theory}

Suppose an agent starts with the uniform prior  about a binary state $\theta \in \{0,1\}$. The agent can flexibly acquire information with a posterior separable cost. That is, they choose any signal space $S$ and any state-contingent distributions over signal realizations $f(\cdot\mid\theta=0) \in \Delta(S)$ and $f(\cdot\mid\theta=1)  \in \Delta(S)$. The agent incurs an attention cost that only depends on the distribution of Bayesian posterior beliefs $\pi \in \Delta(\Delta(\{0,1\}))$ induced by $(S,f)$. This cost is $$C(\pi) = H(1/2) - \mathbb{E}_{p \sim \pi} [H(p)]$$
for some measure of uncertainty $H: [0,1] \rightarrow \mathbb{R}$. 

The agent observes a signal realization $s \sim f(\cdot \mid \theta)$ from their chosen information structure. Then, the agent observes for free an exogenous binary symmetric signal about the state with precision $q$, where $1/2 \le q \le 1$. In the experiment, Phase One participant's guess will play the role of this exogenous signal. Using the realizations of $s$ and the free binary signal, the agent chooses an action $a \in \{0,1\}$ and receives a payoff of $\mathbf{1}_{a = \theta}$. The agent maximizes the expectation of $\mathbf{1}_{a =\theta} - C(\pi)$. The idea is that the agent optimally chooses how to learn about the state in anticipation of later receiving some additional information in the form of the exogenous binary signal. Acquiring more precise information about the state is more costly, and the agent trades off the benefits of information against its cost.

Suppose that $H(\cdot)$ satisfies the following properties:
\begin{enumerate}[(1)]
\item \textbf{Concavity}: $H(\cdot)$ is strictly concave.
\item \textbf{Symmetry}: $H(x)=H(1-x)$ for all $x \in [0,1]$.
\item \textbf{Monotonicity}: $H(\cdot)$ is strictly decreasing on $[\frac12,1]$.
\end{enumerate}

\cite{denti2022posterior} shows that any behavior that can be rationalized by a posterior separable cost of information can be rationalized with a concave measure of uncertainty $H(\cdot)$. We impose the slightly stronger assumption of \emph{strict} concavity to rule out excessive indifferences. Symmetry in our experiment corresponds to the cost of acquiring information not depending on differences between the colors orange and purple, and meaningful violations would require substantial perceptual differences between the two. Monotonicity requires stronger posteriors to be more costly to acquire.

A standard uncertainty measure satisfying the three conditions above is the entropic cost \citep{sims2003implications}, defined by$$H(x)=\eta\cdot(-x\log(x)-(1-x)\log(1-x))$$
for some $\eta>0$. Our results also apply to the drift-diffusion
model with a linear waiting cost, which we now describe. There is a continuous-time
process $X_{t}=\theta\cdot t+\sigma\cdot B_{t}$ where $\sigma>0$
and $B_{t}$ is the standard Brownian motion process. The agent observes
this process until they decide to stop, paying a cost $ct$ for some
$c>0$ if they stop at time $t.$ After stopping, the agent then observes
the exogenous binary signal and takes their binary action $a$.  \cite{morris2019wald}
show that any posterior belief distribution supported in $(0,1)$
with a mean of 1/2 can be generated by some stopping time, and the
expected waiting cost associated with the cheapest way to generate
a given belief distribution is a posterior separable function of the
desired belief distribution. \cite{morris2019wald} provide an explicit expression for
the uncertainty measure $H$ associated with this cost function, and
it is straightforward to verify that this $H$ satisfies the three conditions
above.

If an agent optimally acquires information with a posterior separable cost function $H$ satisfying the three properties above, information acquisition takes a simple form: for all but at most one value of $q$, the agent either acquires a binary symmetric signal of precision $p^*>q$ and follows this signal or acquires no information and follows the exogenous signal.

We give a brief intuition. By posterior separability and symmetry, the agent's problem is equivalent to acquiring an optimal mixture of binary symmetric signals of various (observable) precisions $p$. Given such a signal, the agent follows the exogenous signal if $p<q$ and the endogenous signal if $p>q$. In the former case, the agent optimally chooses $p=\frac12$ to save information acquisition costs. In the latter case, strict concavity implies there is a unique optimal precision $p^*$. This gives the desired information acquisition pattern outside of the unique value $\hat{q}$ of $q$ for which the agent is indifferent between these two options.

This structure implies two sharp predictions:
\begin{proposition}[Attention Substitution]\label{p:post_sep}
    Suppose it is optimal for the agent to acquire $\pi \in \Delta(\Delta(\{0,1\}))$ when the exogenous signal has precision $q$ and $\pi' \in \Delta(\Delta(\{0,1\}))$ when the exogenous signal has precision $q'$. If $q' > q \geq \frac12$, then $\pi' \preceq \pi$ in the Blackwell order.
\end{proposition}
Increasing the precision $q$ of the exogenous signal leads to the agent acquiring (weakly) less information and paying a lower cost of information acquisition. For low $q$, the agent acquires and follows a binary symmetric signal of fixed precision. When $q$ crosses the threshold $q^*$, the agent switches to paying zero attention cost and following the exogenous signal.

\begin{proposition}[No Integration]\label{p:no_integration}
    For all except possibly one value of $q$, any optimal attention choice with a strictly positive cost has the property that the agent's optimal action does not depend on the realization of the exogenous signal. 
\end{proposition}
The result  states that under optimal attention, there should be no integration between the information that the agent acquires through attention and the exogenous signal that the agent sees. Either the agent pays zero cost of information acquisition and fully relies on the exogenous signal, or the agent pays some non-zero amount and always follows their endogenously acquired information. Intuitively, agents do not pay anything for information unless it will fully determine their decision.

\subsection{Results from Attention-Substitution Task}

\subsubsection{Basic Patterns}

We begin by considering how subjects combine their own preliminary guess with a peer's guess. We will then turn in the following subsections to the question of when subjects pay attention.

When a subject's preliminary guess agrees with the peer's guess, they follow these guesses about 99\% of the time in each treatment. The more interesting case is when their preliminary guess disagrees with the peer's guess. Subjects change their initial guess less frequently when it is contradicted by the guess of a low-incentive peer than when it is contradicted by a high-incentive peer. To see this, we restrict attention to the rounds where the subject's initial guess based on the grid differs from the peer's guess revealed on the next screen, and we regress whether the subject changes their final incentivized guess  on the incentive level of the peer (column (3) of Table \ref{tab:anticipation}). We find that subjects change their guesses about 21\% less often when they disagree with a low-incentive peer than a high-incentive peer, a difference significant at the 10\% level. This suggests that after signals have realized, subjects' aggregate behavior does account for the effect of incentives on accuracy.

Subjects' final guesses are also more accurate when their peer has a higher incentive (73\% versus 68\%). Subjects in both treatments are obtaining higher accuracy than the peers' guesses they rely on, and the magnitude of improvement is also comparable across the treatments.

\subsubsection{Subjects Violate Attention Substitution}

We compare the attention choices of subjects who are paired with  high-incentive  peers with the attention choices of those  who are paired with low-incentive  peers. In contrast to the prediction from Proposition \ref{p:post_sep}, we do not find that subjects pay more attention to the grid when they anticipate that they will observe a guess from a peer with a lower incentive treatment (which is a less precise binary signal about the state). Table \ref{tab:anticipation} shows that, measured either in terms of the duration that subjects spent looking at the grid or in terms of their initial guess accuracy (where they could only use what they learned from their own attention),  subjects actually pay \emph{less} attention directionally when they face a peer with a lower incentive level. This difference is even statistically significant at the 5\% level for attention duration. Restricting to subjects who correctly estimate the directional effect of incentive on peers' accuracy (Table \ref{tab:anticipation_directionallycorrect}), we continue to find that subjects pay less attention directionally when facing lower-incentive peers, though the difference is no longer statistically significant.

Why do we see subjects pay less attention when learning from a peer with a lower incentive level? As described above, this finding cannot be explained by models of rational inattention with posterior-separable costs of information acquisition. It is possible to rationalize this behavior with other costs of information acquisition. For example, suppose that agents choose between a low-cost information structure which guarantees a moderately strong posterior (e.g., a single binary, symmetric signal) and a high-cost information structure which can give either a weak posterior or a strong posterior (e.g., two binary, symmetric signals). Then the agent may optimally choose the low-cost information structure with an imprecise exogenous signal and the high-cost information structure with a precise exogenous signal: the possibility of a weak posterior is less concerning when the exogenous signal is more informative. We do not see an obvious reason why the attentional task in our experiment would generate this type of cost structure, however, and it is not clear whether more natural cost structures can rationalize our finding.

An alternate explanation is that subjects' behavior does not match rational inattention models with any information costs. One channel could be cognitive. The effect of incentive on peer's accuracy and the required adjustment to one's strategy may be harder to think through prospectively (before the realization of any information) than retrospectively (after paying attention, having extracted some information through attention and having seen the peer's binary guess). Ex-post, subjects paired with a high-incentive peer may understand that they should weigh the peer's guess more heavily and change their own initial guess when they find that it is contradicted by the peer. But ex-ante, at the moment of choosing how much attention they should pay to the grid, subjects may not be able to reason through the various possible contingencies and understand how their optimal attention choice should depend on the peer's incentive. The effect could also be explained via coordination motives or competitive preferences. If subjects want their unincentivized guesses to agree with Phase One peers, they will have a stronger incentive to acquire information when those peers are more accurate. Similar effects could also arise if subjects want to be more accurate than their Phase One peers. 

A final explanation could come from the experimental protocols. To minimize experimenter demand effects, we exposed subjects to only one incentive treatment until tasks requiring introducing the other treatment. As a result, Phase Two subjects who were paired with high-incentive peers for the attention choice in anticipation of social information task also faced the high incentive in the two initial trial tasks (and similarly for low incentives). It seems plausible that the trial tasks anchor effort levels and that this impacts effort in later parts of the experiment.\footnote{To  minimize such effects, we included a second round of comprehension questions before the Attention-Substitution Task. One of these questions required participants to report the $\$0.50$ bonus payment for correct guesses in this task.} We note that the direction of any interference between parts of the experiment is not obvious, however. In particular, the attention choice in anticipation of social information task has an intermediate $\$0.50$ incentive. This incentive may be more motivating for subjects who could earn $\$0.05$ per trial task than those who could earn $\$1.25$ per trial task.


\subsubsection{Subjects Violate No Integration}

While we cannot directly observe whether subjects incurred a non-zero attention cost, we can fairly confidently classify subjects who spent at least ten seconds on the grid as having paid some cost of attention, for they could have saved time by immediately advancing to the next screen. Proposition \ref{p:no_integration} predicts that these subjects should never change their initial guess, even when it is contradicted by the peer's answer on the next screen. By contrast, we find that around 20\% of these subjects switch their initial guess when the peer's answer disagrees with it. This rate does not change much between the first five repetitions and last five repetitions of the task (20.7\% versus 20.1\%), and is also similar among the subset of subjects who correctly identify the directional effect
of incentives on accuracy (19.3\%).

We also ask how many of the subjects who switch guesses spent substantial time on the attention task. We find that among the subjects who switch their initial guess when contradicted by the peer's answer, about 44\% spent ten or more seconds on the grid. This proportion is 53.1\% in the first five repetitions of the task and 36.1\% in the last five repetitions of the task, which may relate to subjects spending less time paying attention to the grid in later repetitions. 

This behavior violates models of rational inattention with posterior-separable costs, but is easier to rationalize with other attention costs. To explain why, suppose that an agent acquires information dynamically (e.g., as in the drift-diffusion model described in Section \ref{s:theory}) but the costs of information increase over time. Then it can be rational to acquire information for some time but then stop and follow the exogenous belief if a strong posterior is not reached. In the experiment, this could correspond to attempting to count the colored balls but giving up and following the peer's signal (rather than trying to count again) if this counting does not yield a clear answer. Beyond the drift-diffusion model, we conjecture that violations of No Integration can be rationalized within recently introduced families of cost functions generalizing posterior separability (iteratively differentiable cost functions, defined by \cite{lipnowski2022predicting}, and $f$-information, defined by \cite*{bloedel2025modeling}).

\section{Concluding Discussion}

This paper studies a social-learning setting where people rely on intermediaries who choose how much attention to pay to their available information, with these attention choices shaped by their incentives. In our experiment, variations in incentives generate meaningful differences in subjects' accuracy, but few subjects consistently understand even the direction of these incentive-driven differences. Many subjects do not respond to others' incentives for accuracy in a variety of tasks: estimation tasks where subjects directly predict  others' accuracy, aggregation tasks where subjects must combine conflicting answers from multiple peers, and social-learning tasks where subjects decide how much attention to pay themselves in anticipation of receiving social information from a peer. The misunderstanding of incentive effects documented in our experiments may cause learners to underweight high-quality sources and potentially to misallocate their own attention.

\bibliography{attention}

\appendix

\section{Proofs}

This section will prove \Cref{p:post_sep} and \Cref{p:no_integration}. We begin with a lemma characterizing optimal posteriors.

\begin{lemma}
There exists $p^*$ independent of $q$ and a threshold $\hat{q}$ such that the optimal information structure is a binary symmetric signal of precision $p^*$ when $q < \hat{q}$, mixes between acquiring a binary symmetric signal of precision $p^*$ and acquiring no information when $q=\hat{q}$, and acquires no information when $q>\hat{q}$.
\end{lemma}
\begin{proof}
Suppose we are given $H:[0,1]\to\mathbb{R}$  satisfying properties (1)-(3) from Section~\ref{s:theory}.
Then $C$ satisfies the following properties:
\begin{enumerate}[(i)]
\item $C(1/2)=0,$
\item $C$ is symmetric around 1/2, and
\item $C$ is strictly increasing and strictly convex on $[1/2,1].$ 
\end{enumerate}

Consider first the problem of maximizing expected utility over information
structures that induce a symmetric distribution of posteriors around
1/2. The expected payoff given a distribution of posteriors $P\in\Delta([0,1])$
is $\mathbb{E}_{p\sim P}[\max(q,\max(p,1-p))-C(p)]$. Let $p^{*}$
solve $\max_{p\in[1/2,1]}\{p-C(p)\}.$ Since $p-C(p)$ is strictly
concave, there is a unique maximizer $p^{*}$ in $[1/2,1].$ 

Suppose $q\le p^{*}-C(p^{*})$. Then it is suboptimal for $P$ to put positive
weight on $(1/2,q]$, and also suboptimal to put positive weight on
any $p\in(q,1]$ with $p\ne p^{*}$, since for such $p$ we would
have $\max(q,\max(p,1-p))-C(p)=p-C(p)<p^{*}-C(p^{*}).$ If the inequality $q\le p^{*}-C(p^{*})$
is strict, the only solution is
for $P$ to put probability 1/2 on $p^{*}$ and 1/2 on $1-p^{*}$.
If we have $q=p^{*}-C(p^{*})$, then the solutions are for $P$ to
put some probability on 1/2 and the rest divided evenly on $p^{*}$
and $1-p^{*}$. 

Suppose $q>p^{*}-C(p^{*})$. Then again it is suboptimal for $P$
to put any weight on $(1/2,q]$, and it is also suboptimal to put
any weight on any $p>q$, because $\max(q,\max(p,1-p))-C(p)=p-C(p)\le p^{*}-C(p^{*})<q$.
So the only solution is for $P$ to put probability 1 on 1/2. We have shown that in the constrained problem where the agent only maximizes over
the class of information structures that induce a symmetric distribution
of posteriors around 1/2, the statement of the lemma holds with $\hat{q} = p^*-C(p^*)$.

Now suppose the agent is also allowed to choose information structures
that induce asymmetric distributions of posteriors around 1/2. Given
any asymmetric distribution $\tilde{P}$ of posteriors, consider the
feasible, symmetric distribution of posterior $P$ defined by $P(p)=\frac{1}{2}\tilde{P}(p)+\frac{1}{2}\tilde{P}(1-p).$
By the symmetry of $H$ and the symmetry of the decision problem,
$P$ and $\tilde{P}$ give the same expected utility. Suppose that $q<p^{*}-C(p^{*})$. If there is an asymmetric $\tilde{P}$ inducing a distribution over posteriors not equal to probability 1/2 on $p^{*}$ and probability 1/2 on $1-p^{*}$,
the constructed $P$ would be a symmetric solution to the agent's
problem which puts positive probability on some posterior belief other
than $p^{*}$ in $[1/2,1].$ This contradicts the uniqueness of the
solution among the class of symmetric belief distributions. The similar
argument also rules out asymmetric solutions in the cases of $q=p^{*}-C(p^{*})$
and $q>p^{*}-C(p^{*})$. Therefore there are no strictly asymmetric
solutions to the agent's problem. 
\end{proof}

The proofs of the propositions are easy corollaries.

\begin{proof}[Proof of \Cref{p:post_sep}]
If $q > \hat{q},$ then the agent acquires no information when the exogenous signal has precision $q$ and so the result is immediate.

If $q = \hat{q}$, the agent mixes between acquiring no information and acquiring a binary signal of precision $p^*$. For any $q' > \hat{q}$, the agent acquires no information.

If $q < \hat{q}$, the agent acquires binary signals of precision $p^*$ with exogenous signals of precisions $q$. With precision $q' > q$, the agent acquires some mixture of a binary signal of precision $p^*$ and no information. So the information structure with precision $q$ is weakly more Blackwell informative than the information structure with precision $q'$.
\end{proof}

\begin{proof}[Proof of \Cref{p:no_integration}]
If $q > \hat{q}$, the information cost is zero. If $q < \hat{q}$, the agent chooses an action determined by the realization of the signal they acquire, so the action does not depend on the exogenous signal. The remaining value $q=\hat{q}$ is the possible exception in the statement.
\end{proof}

\section{Additional Tables}

\begin{table}[h!]
\centering
\begin{talltblr}[         
caption={Time spent on Phase One tasks},
label={tab:reg_duration},
note{}={Standard errors clustered at the subject level.},
]                     
{                     
colspec={Q[]Q[]},
column{2}={}{halign=c,},
column{1}={}{halign=l,},
hline{8}={1-2}{solid, black, 0.05em},
}                     
\toprule
& Duration \\ \midrule 
Constant & \num{28.27}*** \\
& (\num{1.35}) \\
Round & \num{-0.95}*** \\
& (\num{0.05}) \\
High Incentive & \num{12.07}*** \\
& (\num{1.96}) \\
N & \num{11938} \\
$R^2$ & \num{0.07} \\
Adj. $R^2$ & \num{0.07} \\
\bottomrule
\end{talltblr}
\end{table}

\begin{table} 
\centering
\begin{talltblr}[         
caption={Regressions of accuracy in Phase One attention tasks.},
 label = {tab:accuracy_five_specs},
note{}={SEs clustered at the subject level.},
]                     
{                     
colspec={Q[]Q[]Q[]Q[]Q[]Q[]},
column{2-6}={}{halign=c,},
column{1}={}{halign=l,},
}                     
\toprule
 & \SetCell[c=5]{c} Accuracy in Phase One Attention Tasks \\
& (1) & (2) & (3) & (4) & (5) \\ \midrule 
Constant & \num{0.629}*** & \num{0.627}*** & \num{0.538}*** & \num{0.503}*** &  \\
& (\num{0.009}) & (\num{0.012}) & (\num{0.011}) & (\num{0.012}) &  \\
High incentive & \num{0.060}*** & \num{0.058}*** & \num{0.020}* & \num{0.013} &  \\
& (\num{0.014}) & (\num{0.014}) & (\num{0.011}) & (\num{0.011}) &  \\
Duration &  &  & \num{0.003}*** & \num{0.006}*** & \num{0.001}*** \\
&  &  & (\num{0.000}) & (\num{0.000}) & (\num{0.000}) \\
Duration$^2$ &  &  &  & \num{-0.000}*** &  \\
&  &  &  & (\num{0.000}) &  \\
Round &  & \num{0.000} & \num{0.003}*** & \num{0.004}*** & \num{0.001} \\
&  & (\num{0.001}) & (\num{0.001}) & (\num{0.001}) & (\num{0.001}) \\
\midrule
N & \num{12000} & \num{11938} & \num{11938} & \num{11938} & \num{11938} \\
$R^2$ & \num{0.004} & \num{0.004} & \num{0.044} & \num{0.052} & \num{0.703} \\
Adj. $R^2$ & \num{0.004} & \num{0.004} & \num{0.044} & \num{0.052} & \num{0.687} \\
Subject FE & No & No & No & No & Yes \\
\bottomrule
\end{talltblr}
\end{table}

\begin{table} 
\centering
\begin{talltblr}[         
caption={Phase Two subjects' estimates of Phase One guess accuracies},
label = {tab:point_estimates},
note{}={(1): Only using subjects' estimates about peers in the same incentive treatment.},
note{ }={(2): Using subjects' estimates about peers in both treatment groups, SEs clustered at the subject level.},
]                     
{                     
colspec={Q[]Q[]Q[]},
column{2-3}={}{halign=c,},
column{1}={}{halign=l,},
hline{6}={1-3}{solid, black, 0.05em},
}                     
\toprule
& (1) & (2)  \\ \midrule 
Constant & \num{0.567}*** & \num{0.562}*** \\
& (\num{0.012}) & (\num{0.009}) \\
High incentive & \num{0.037}* & \num{0.057}*** \\
& (\num{0.019}) & (\num{0.008}) \\
N & \num{393} & \num{786} \\
$R^2$ & \num{0.010} & \num{0.022} \\
Adj. $R^2$ & \num{0.007} & \num{0.021} \\
\bottomrule
\end{talltblr}
\end{table}

\begin{table}
\centering
\begin{talltblr}[         
caption={Estimated Beliefs about Log-likelihoods},
label  ={tab:log_belief_regression},
note{}={Standard errors clustered at the subject level. To make belief log-likelihood ratio well defined, beliefs of 0 and 1 are changed to 0.001 and 0.999.},
]                     
{                     
colspec={Q[]Q[]},
column{2}={}{halign=c,},
column{1}={}{halign=l,},
hline{6}={1-2}{solid, black, 0.05em},
}                     
\toprule
& Belief Log-likelihood \\ \midrule 
Low Incentive Log-likelihood & \num{0.257}*** \\
& (\num{0.030}) \\
High Incentive Log-likelihood & \num{0.489}*** \\
& (\num{0.049}) \\
N & \num{1572} \\
$R^2$ & \num{0.12} \\
Adj. $R^2$ & \num{0.12} \\
\bottomrule
\end{talltblr}
\end{table}

\begin{table}
\centering
\begin{talltblr}[         
caption={Effect of Peer Incentive on Own Attention and Guess Aggregation},
label = {tab:anticipation},
note{}={SEs clustered at the subject level.},
note{ }={Columns (1) and (2) exclude rounds where subject spent more than 225 seconds on the grid.},
]                     
{                     
colspec={Q[]Q[]Q[]Q[]},
column{2-4}={}{halign=c,},
column{1}={}{halign=l,},
hline{6}={1-4}{solid, black, 0.05em},
}                     
\toprule
& (1) Attention Duration (s) & (2) Initial Guess Accuracy & (3) Prob of Changing Guess \\ \midrule 
Constant & \num{25.936}*** & \num{0.676}*** & \num{0.220}*** \\
& (\num{2.007}) & (\num{0.014}) & (\num{0.022}) \\
High Peer Incentive & \num{6.443}** & \num{0.008} & \num{0.057}* \\
& (\num{2.953}) & (\num{0.021}) & (\num{0.034}) \\
N & \num{3915} & \num{3915} & \num{1709} \\
$R^2$ & \num{0.009} & \num{0.000} & \num{0.004} \\
Adj. $R^2$ & \num{0.008} & \num{-0.000} & \num{0.004} \\
\bottomrule
\end{talltblr}
\end{table}

\begin{table}
\centering
\begin{talltblr}[         
caption={Effect of Peer Incentive on Own Attention and Guess Aggregation, Among Subjects with Directionally Correct Estimates of the Effect of Incentive on Accuracy},
label = {tab:anticipation_directionallycorrect},
note{}={SEs clustered at the subject level.},
note{ }={Columns (1) and (2) exclude rounds where subject spent more than 225 seconds on the grid.},
]                     
{                     
colspec={Q[]Q[]Q[]Q[]},
column{2-4}={}{halign=c,},
column{1}={}{halign=l,},
hline{6}={1-4}{solid, black, 0.05em},
}                     
\toprule
& (1) Attention Duration (s) & (2) Initial Guess Accuracy & (3) Prob of Changing Guess \\ \midrule 
Constant & \num{29.205}*** & \num{0.697}*** & \num{0.234}*** \\
& (\num{2.925}) & (\num{0.019}) & (\num{0.032}) \\
High Peer Incentive & \num{6.076} & \num{0.010} & \num{0.031} \\
& (\num{4.022}) & (\num{0.028}) & (\num{0.047}) \\
N & \num{2163} & \num{2163} & \num{915} \\
$R^2$ & \num{0.007} & \num{0.000} & \num{0.001} \\
Adj. $R^2$ & \num{0.007} & \num{-0.000} & \num{0.000} \\
\bottomrule
\end{talltblr}
\end{table}

\begin{table}
\centering
\begin{talltblr}[         
caption={Determinants of Phase Two Subjects Making Directionally Correct Judgments},
label={tab:who_understands},
note{}={Each column is a separate univariate regression of the column outcome on a single regressor. Standard errors (in parentheses) are heteroskedasticity-robust (HC1). Attention Task Accuracy is average accuracy in the two trial attention tasks and in the initial guesses for the ten repetitions of the attention substitution task. Attention Task Duration is the average number of seconds spent on these attention tasks. Judgment Task Duration is the average number of seconds spent on the forecasting others' accuracy under different incentive levels for columns (1) to (3) or the number of seconds spent on aggregating the two Phase One subjects' guesses for columns (4) to (6).},
]                     
{                     
colspec={Q[]Q[]Q[]Q[]Q[]Q[]Q[]},
column{2-7}={}{halign=c,},
column{1}={}{halign=l,},
hline{11}={1-7}{solid, black, 0.05em},
}                     
\toprule
 & \SetCell[c=3]{c} Correct Estimate of Incentive Effect Sign & & & \SetCell[c=3]{c} Directionally Correct Aggregations & & \\
\cmidrule[lr]{2-4}\cmidrule[lr]{5-7}
& (1) & (2) & (3) & (4) & (5) & (6) \\ \midrule 
Constant & \num{0.302}*** & \num{0.508}*** & \num{0.509}*** & \num{0.065} & \num{0.197}*** & \num{0.177}*** \\
& (\num{0.092}) & (\num{0.037}) & (\num{0.039}) & (\num{0.075}) & (\num{0.028}) & (\num{0.032}) \\
Attention Task Accuracy & \num{0.368}*** &  &  & \num{0.222}** &  &  \\
& (\num{0.128}) &  &  & (\num{0.111}) &  &  \\
Attention Task Duration &  & \num{0.001} &  &  & \num{0.001} &  \\
&  & (\num{0.001}) &  &  & (\num{0.001}) &  \\
Judgment Task Duration &  &  & \num{0.001} &  &  & \num{0.000} \\
&  &  & (\num{0.001}) &  &  & (\num{0.000}) \\
N & \num{393} & \num{393} & \num{393} & \num{393} & \num{393} & \num{393} \\
$R^2$ & \num{0.020} & \num{0.007} & \num{0.005} & \num{0.011} & \num{0.002} & \num{0.006} \\
Adj. $R^2$ & \num{0.018} & \num{0.005} & \num{0.003} & \num{0.008} & \num{-0.001} & \num{0.004} \\
\bottomrule
\end{talltblr}
\end{table}

\clearpage
\section{Experimental Instructions and Interface}\label{sec:interface}

\begin{figure}[H]
    \centering
    \includegraphics[width=0.6\linewidth]{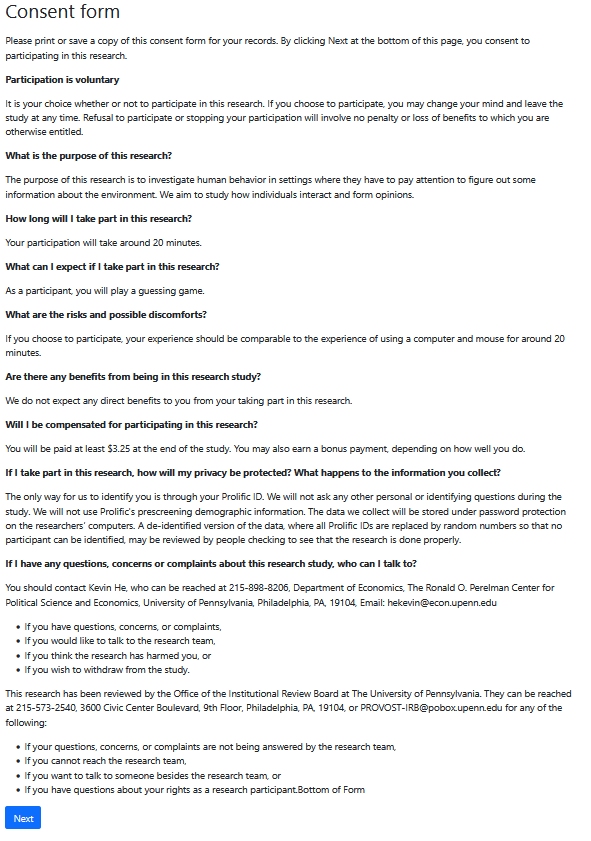}
    \caption{Consent form.}
    \label{fig:w1s1}
\end{figure}

\begin{figure}[H]
    \centering
    \includegraphics[width=0.9\linewidth]{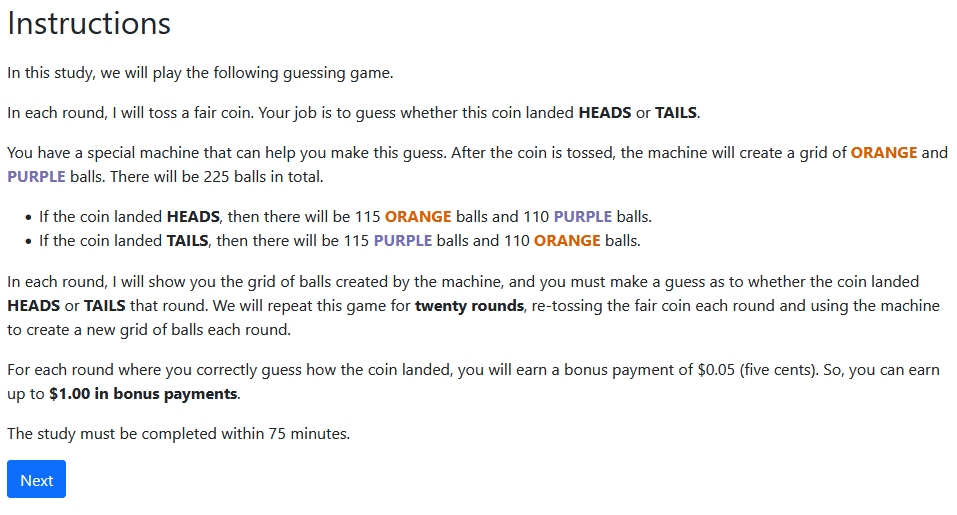}
    \caption{Instructions for Phase One subjects. Half of the subjects are in the low-incentive treatment (shown here) and the other half of the subjects are in the high-incentive treatment and receive \$1.25 for each round where they correctly guess how the coin landed.}
    \label{fig:w1s2}
\end{figure}

\begin{figure}
    \centering
    \includegraphics[width=0.9\linewidth]{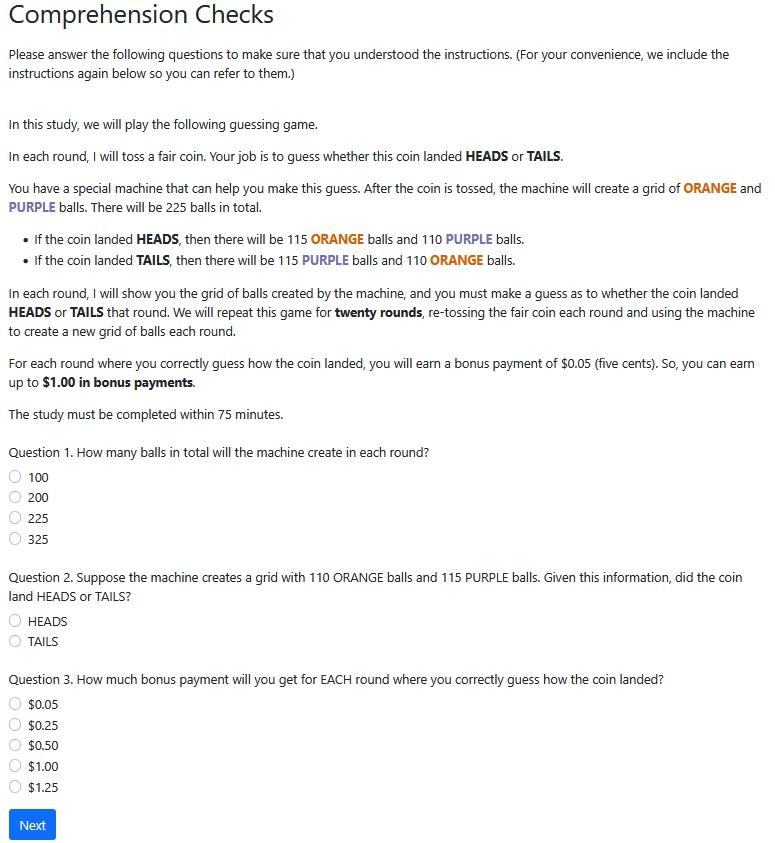}
    \caption{Comprehension questions for Phase One subjects.}
    \label{fig:w1s3}
\end{figure}

\begin{figure}
    \centering
    \includegraphics[width=0.9\linewidth]{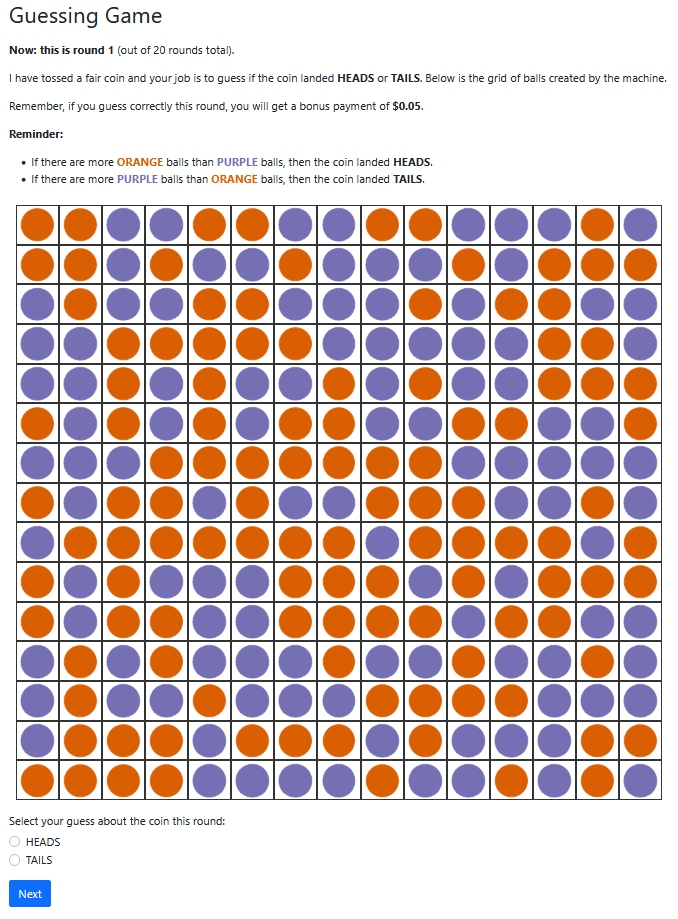}
    \caption{Attention task for Phase One subjects. Subjects completed 20 repetitions of this task.}
    \label{fig:w1s4}
\end{figure}

\begin{figure}
    \centering
    \includegraphics[width=1\linewidth]{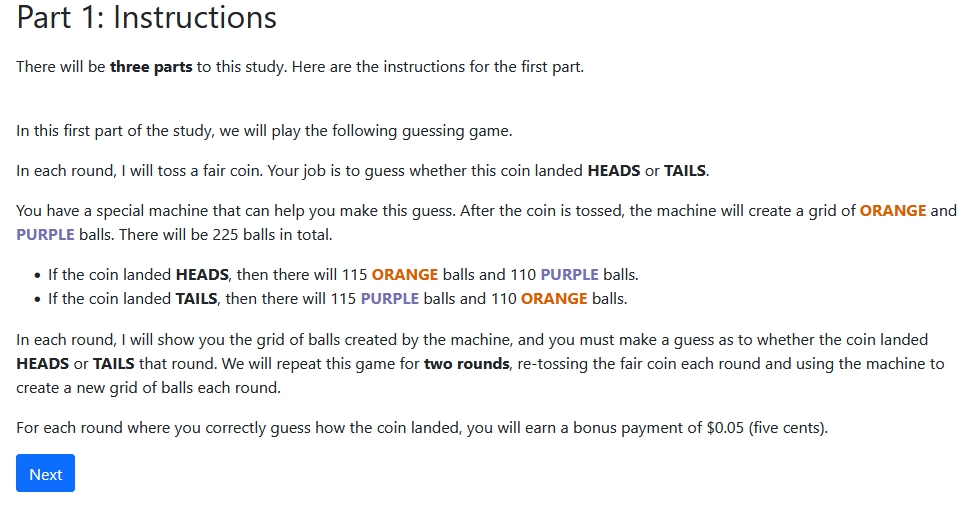}
    \caption{Instructions for Part 1 of the study for Phase Two subjects. Half of the subjects are in the low-incentive treatment (shown here) and the other half are in the high-incentive treatment and receive \$1.25 for each round where they correctly guess how the coin landed.}
    \label{fig:w2s1}
\end{figure}

\begin{figure}
    \centering
    \includegraphics[width=1\linewidth]{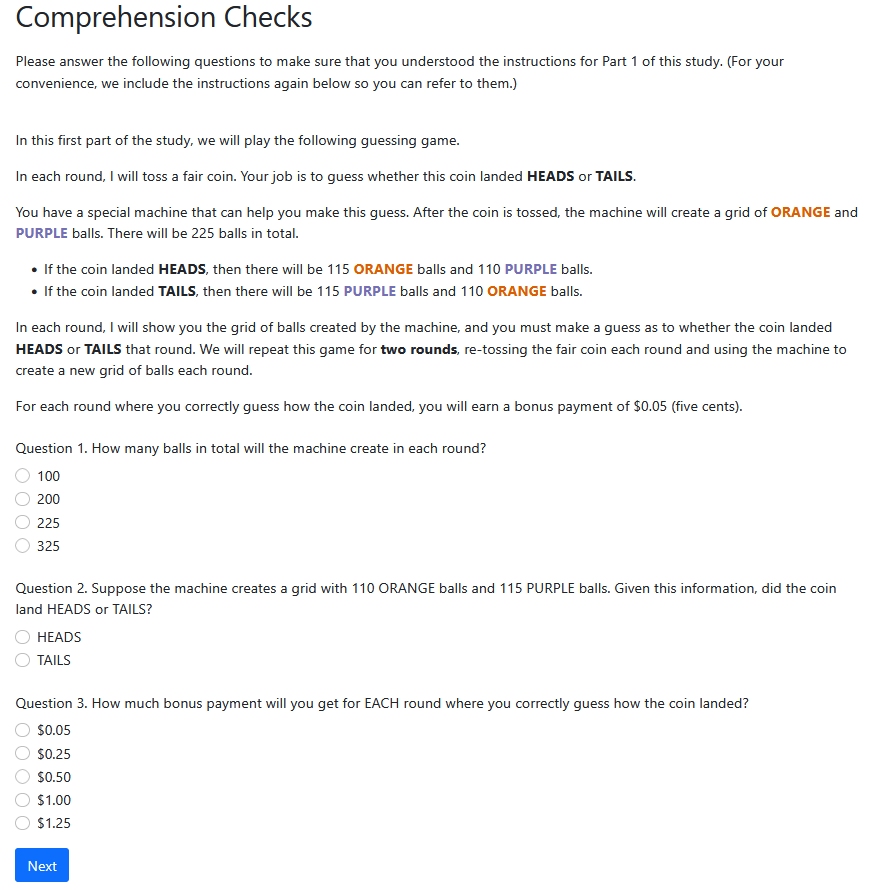}
    \caption{Comprehension questions for Part 1 of Phase Two.}
    \label{fig:w2s2}
\end{figure}

\begin{figure}
    \centering
    \includegraphics[width=0.9\linewidth]{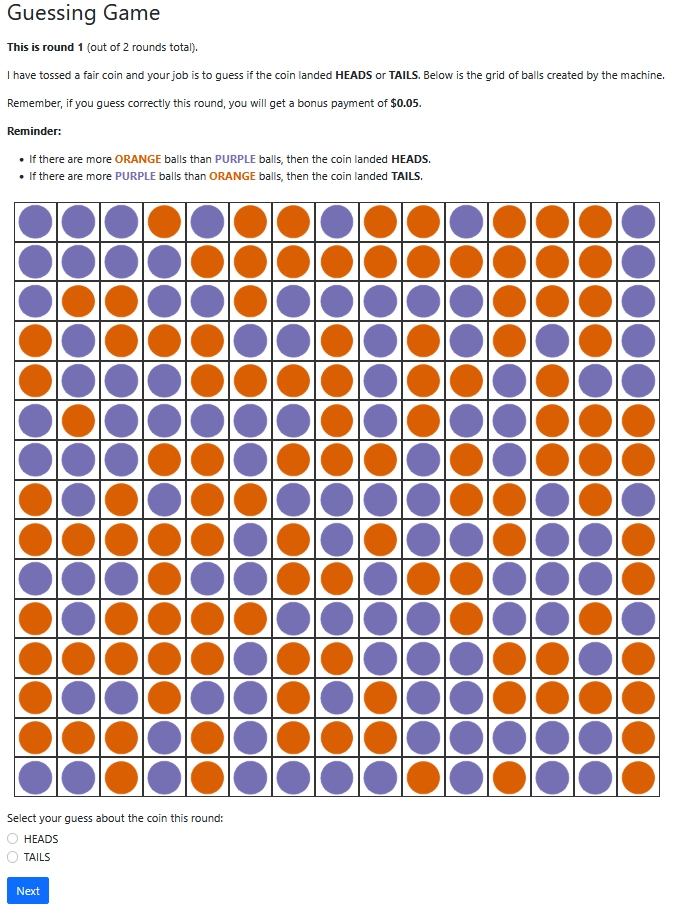}
    \caption{Attention tasks for Part 1 of Phase Two.}
    \label{fig:w2s3}
\end{figure}

\begin{figure}
    \centering
    \includegraphics[width=1\linewidth]{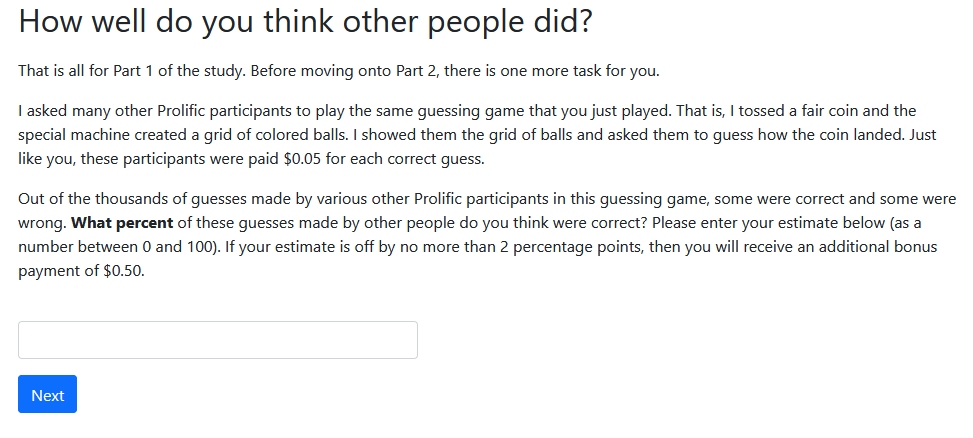}
    \caption{Asking Phase Two subjects to estimate the accuracy of Phase One subjects with the same incentive level. (This screen shows a Phase Two subject in the low-incentive treatment.)}
    \label{fig:w2s4}
\end{figure}

\begin{figure}
    \centering
    \includegraphics[width=1\linewidth]{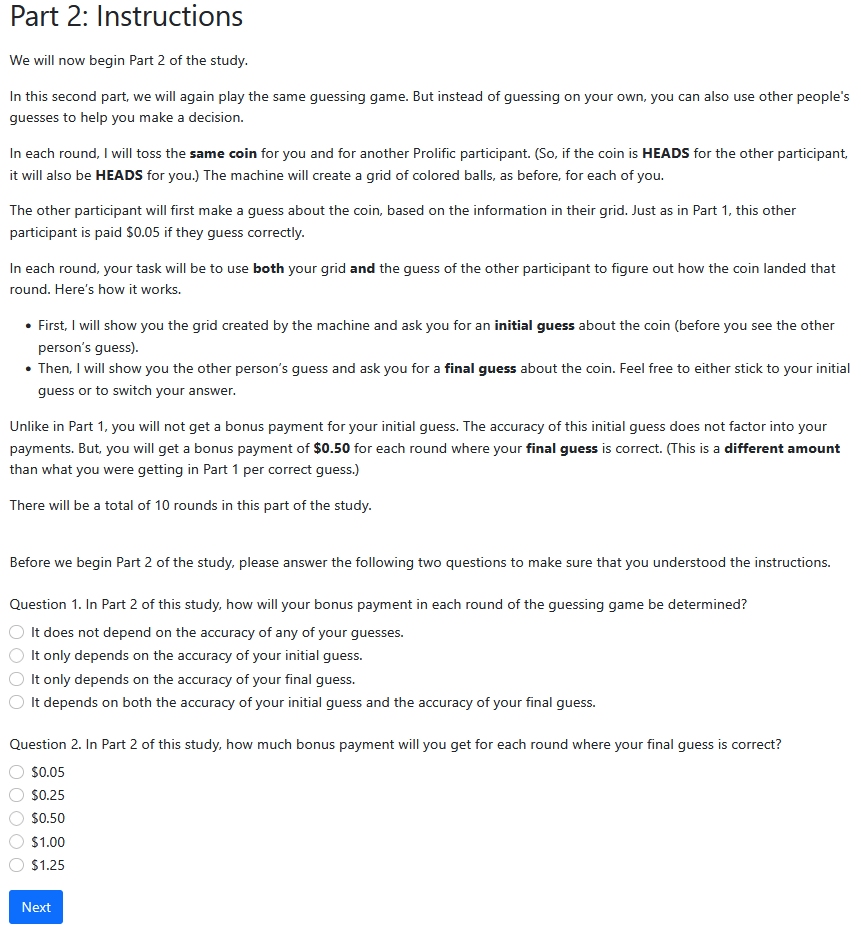}
    \caption{Instructions and comprehension questions for Part 2 of Phase Two. All Phase Two subjects are paid \$0.50 for each correct final guess, but subjects in low-incentive and high-incentive treatments saw the guesses of predecessors with different incentive levels.}
    \label{fig:w2s5}
\end{figure}

\begin{figure}
    \centering
    \includegraphics[width=0.8\linewidth]{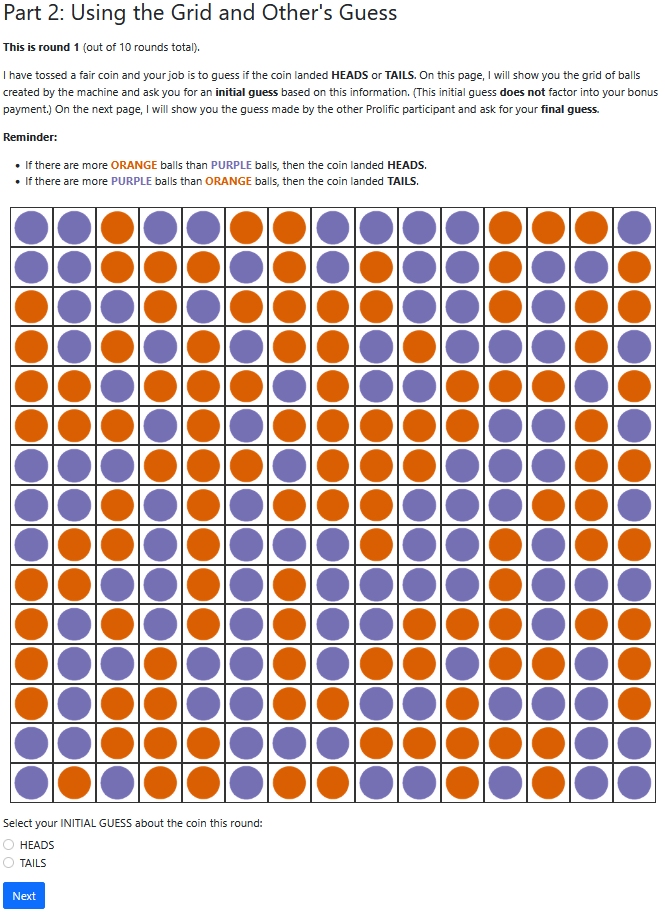}
    \caption{The first screen of the Attention-Substitution Task. The subjects see the grid and must make an initial guess to move on to the next screen.}
    \label{fig:w2s6}
\end{figure}

\begin{figure}
    \centering
    \includegraphics[width=0.8\linewidth]{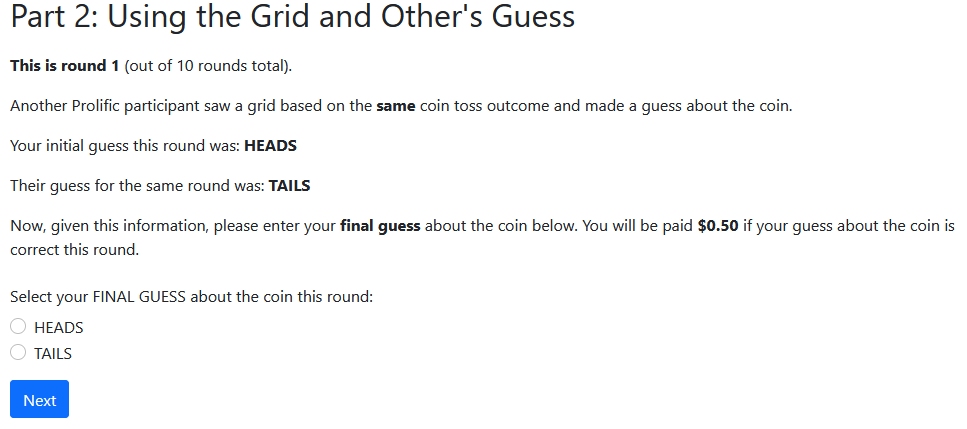}
    \caption{The second screen of the Attention-Substitution Task. The subjects see their initial guess and the guess of a Phase One subject. They cannot return to the previous screen with the grid and must make a final guess. Subjects repeated this process for ten rounds.}
    \label{fig:w2s7} 
\end{figure}

\begin{figure}
    \centering
    \includegraphics[width=0.8\linewidth]{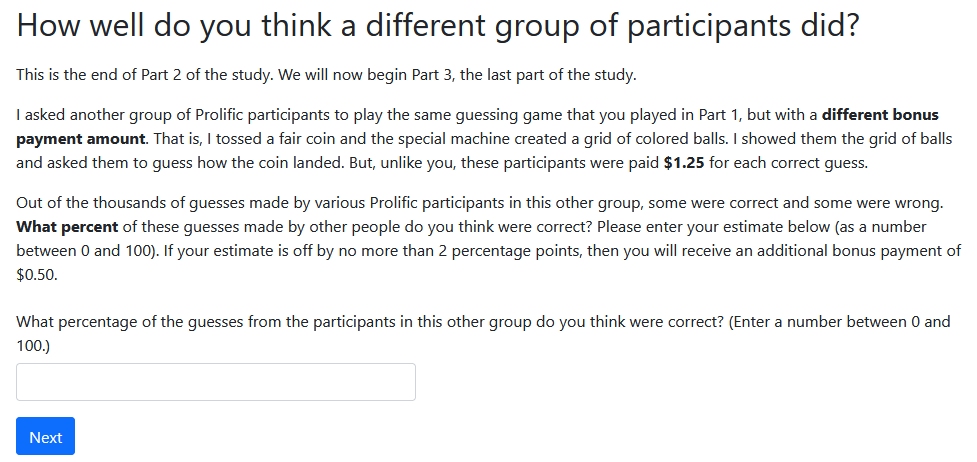}
    \caption{Asking Phase Two subjects to estimate the accuracy of Phase One subjects with  a different incentive level. (This screen shows a Phase Two subject in the low-incentive treatment.)}
    \label{fig:w2s8} 
\end{figure}

\begin{figure}
    \centering
    \includegraphics[width=0.8\linewidth]{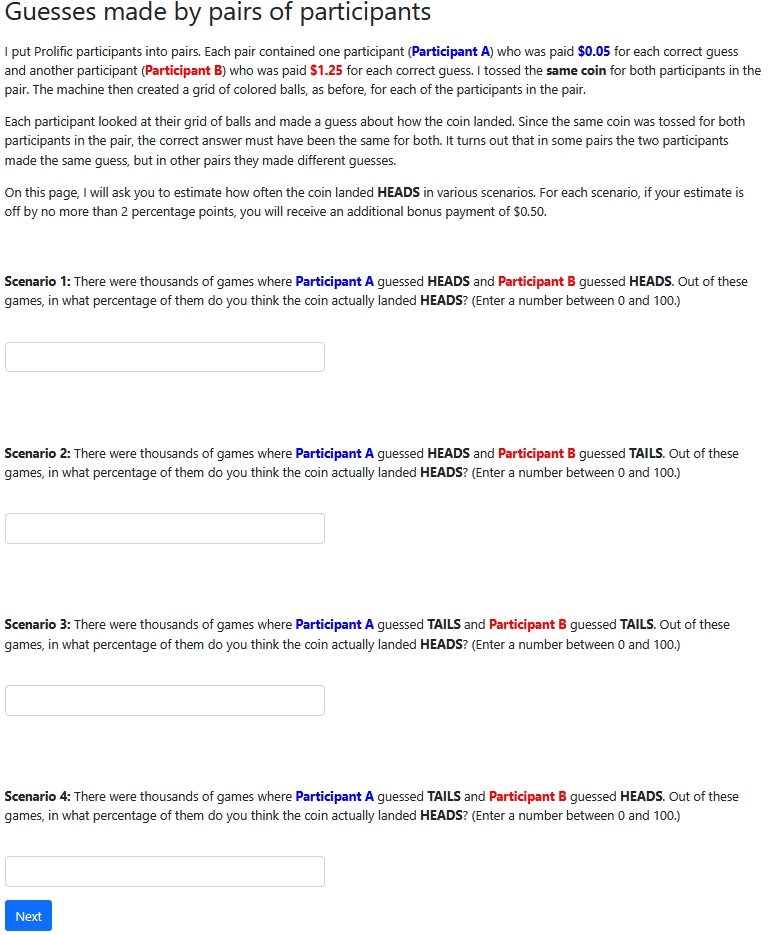}
    \caption{Asking Phase Two subjects to aggregate guesses made by two Phase One subjects from different incentive treatments. We randomized the phrasing of this question across subjects, so ``Participant A'' is in the low-incentive treatment for half of the Phase Two subjects and ``Participant B'' is in the low-incentive treatment for the other half. We also randomized the order of the four scenarios that correspond to guesses  HH, HT, TH, and TT from Participant A and Participant B.}
    \label{fig:w2s9} 
\end{figure}

\end{document}